\documentclass[a4paper]{article}
\usepackage{fullpage}
\usepackage{xcolor}
\usepackage{amsmath}
\usepackage{hyperref}
\usepackage[natbib,theorems,cref,color]{nikis}
\newcommand*{\inst}[1]{}
\newcommand*{\authorrunning}[1]{}
\newcommand*{\institute}[1]{}
\newcommand*{\keywords}[1]{Keywords: #1}

\usepackage{microtype}
\date{}

\usepackage{ifthen}
\newboolean{isJournal} 
\setboolean{isJournal}{true} 
\ifthenelse{\boolean{isJournal}}{%
    \def\Jvar{}
    
    \newcommand{\J}[2][]{#2}
    \newcommand{\JO}[2][]{#2}
    \newcommand{\JProof}[2]{\begin{proof}#2\end{proof}}

    \newenvironment{Jfigure}[1][]{\begin{figure}[#1]}{\end{figure}}
    \newenvironment{Jwrapfigure}[3][]{\begin{wrapfigure}[#1]{#2}{#3}}{\end{wrapfigure}}

}{%
    \usepackage{etoolbox}
    \def\Jvar{}
    
    \newcommand{\J}[2][]{} %
    \newcommand{\JO}[2][]{#1}
    \newcommand{\JProof}[2]{\J{\begin{proof}[of {#1}]#2\end{proof}}}

}

\usepackage{booktabs}
\usepackage{tabularx}
\usepackage{enumitem}

\usepackage{arydshln}
\makeatletter
\def\adl@drawiv#1#2#3{%
        \hskip.5\tabcolsep
        \xleaders#3{#2.5\@tempdimb #1{1}#2.5\@tempdimb}%
                #2\z@ plus1fil minus1fil\relax
        \hskip.5\tabcolsep}
\newcommand{\cdashlinelr}[1]{%
  \noalign{\vskip\aboverulesep
           \global\let\@dashdrawstore\adl@draw
           \global\let\adl@draw\adl@drawiv}
  \cdashline{#1}
  \noalign{\global\let\adl@draw\@dashdrawstore
           \vskip\belowrulesep}}
\makeatother

\usepackage{graphicx}
\usepackage[utf8]{inputenc}

\usepackage[final,mode=multiuser]{fixme}
\fxuseenvlayout{colorsig}
\fxusetargetlayout{color}
\FXRegisterAuthor{ti}{ati}{TI}
\FXRegisterAuthor{dk}{adk}{DK}
\FXRegisterAuthor{kg}{akg}{KG}
\definecolor{kgnote}{rgb}{1.0000,0.0000,0.0000}
\fxsetface{env}{}

\newcommand*{\teigi}[1]{{\emph{#1}}}

\DeclareMathOperator{\fnDeletePrefixRun}{rmPreRun}
\DeclareMathOperator{\fnDeleteSuffixRun}{rmSufRun}
\DeclareMathOperator{\fnPopcount}{popcount}
\DeclareMathOperator{\fnFind}{find}
\DeclareMathOperator{\fnFindBigram}{findBigram}
\DeclareMathOperator{\fnMSB}{msb}
\newcommand*{\fnAND}{\mathbin{\&}}
\newcommand*{\fnOR}{\mathbin{|}}
\newcommand*{\fnXOR}{\mathbin{\otimes}}

\DeclareMathOperator{\sort}{sort}
\DeclareMathOperator{\scan}{scan}

\begin{document}
\JO[\pagestyle{fancyplain}]{}%

\JO[%
\title{Re-Pair in Small Space$^\ast$%
\footnote{$^\ast$The full version of this paper is available at \texttt{http://arxiv.org/abs/1908.04933} with a practical evaluation and improved running times in the word-packed model or for CRCW machines.}
}]{\title{Re-Pair in Small Space}}

\authorrunning{D. K\"oppl et al.}
\author{%
  Dominik~K\"oppl\inst{1} \and
  Tomohiro I\inst{2} \and
  Isamu Furuya\inst{3} \and
Yoshimasa Takabatake\inst{2} \and
Kensuke Sakai\inst{2} \and
Keisuke Goto\inst{4} \and
\and
\JO[%
{\small\begin{minipage}{\linewidth}\begin{center}
    \inst{1}{Kyushu University/JSPS, Japan,}
  \email{dominik.koeppl@inf.kyushu-u.ac.jp}
  \\
  \inst{2}Kyushu Institute of Technology, Japan,\\
  \{\texttt{tomohiro@}, \texttt{takabatake@}, \texttt{k\_sakai@donald.}\}\texttt{ai.kyutech.ac.jp}
  \\
  \inst{3}Graduate School of IST, Hokkaido University, Japan,
  \email{furuya@ist.hokudai.ac.jp}
  \\
  \inst{4}Fujitsu Laboratories Ltd., Kawasaki, Japan,
\email{goto.keisuke@fujitsu.com}
    \end{center}\end{minipage}}
]{}%
}
\institute{%
   {Kyushu University, Japan Society for Promotion of Science}\\
  \email{dominik.koeppl@inf.kyushu-u.ac.jp},
   \and
   Kyushu Institute of Technology, Japan\\
  \email{tomohiro@ai.kyutech.ac.jp}, \email{takabatake@ai.kyutech.ac.jp}, \email{k\_sakai@donald.ai.kyutech.ac.jp}
  \and
  Graduate School of IST, Hokkaido University, Japan\\
  \email{furuya@ist.hokudai.ac.jp}
\and
Fujitsu Laboratories Ltd., Kawasaki, Japan \\
\email{goto.keisuke@fujitsu.com}
  }

\newcommand*{\Acknowledgments}{%
This work is funded by the JSPS KAKENHI Grant Numbers JP18F18120 (Dominik~K\"oppl), 19K20213 (Tomohiro~I) and 18K18111 (Yoshimasa~Takabatake), and the JST CREST Grant Number JPMJCR1402 including AIP challenge program (Keisuke~Goto).
}

\maketitle
\begin{abstract}
   Re-Pair is a grammar compression scheme with favorably good compression rates. 
The computation of Re-Pair comes with the cost of maintaining large frequency tables, which 
makes it hard to compute Re-Pair on large scale data sets.
As a solution for this problem we present,
given a text of length~$n$ whose characters are drawn from an integer alphabet, 
an \JO[\Oh{n^2}]{$\Oh{n^2}\cap \Oh{n^2 \lg \log_\tau n  \lg \lg \lg n / \log_\tau n}$} time algorithm computing Re-Pair in $n \upgauss{\lg \max(n, \tau)}$ bits of working space including the text space, 
where $\tau$ is the number of terminals and non-terminals.
The algorithm works in the restore model, supporting the recovery of the original input in the time for the Re-Pair computation with \Oh{\lg n} additional bits of working space.
\JO{We give variants of our solution working in parallel or in the external memory model.}
\end{abstract}


\section{Introduction}
Re-Pair~\cite{larsson99repair} is a grammar deriving a single string. 
It is computed by replacing the most frequent bigram in this string with a new non-terminal, recursing until no bigram occurs more than once.
Despite this simple-looking description, both the merits and the computational complexity of Re-Pair are intriguing.
As a matter of fact, Re-Pair is currently one of the most well-understood grammar schemes. 

Besides the seminal work of \citet{larsson99repair}, there are a couple of articles devoted to the compression aspects of Re-Pair:
Given a text~$T$ of length~$n$ whose characters are drawn from an integer alphabet of size~$\sigma$,
the output of Re-Pair applied to~$T$ is at most $2n H_k(T) + \oh{n \lg \sigma}$ bits with $k = \oh{\log_\sigma n}$ when represented naively as a list of character pairs~\cite{navarro08repair},
where $H_k$ denotes the empirical entropy of the $k$-th order.
Using the encoding of~\citet{kieffer00code}, 
\citet{ochoa19repair} could improve the output size to at most $n H_k(T) + \oh{n \lg \sigma}$ bits.
Other encodings were recently studied by \citet{ganczorz19entropy}.
Since Re-Pair is a so-called \emph{irreducible} grammar,
its grammar size, i.e., the sum of the symbols on the right hand of all rules, 
is upper bounded by \Oh{n / \log_\sigma n}~\cite[Lemma~2]{kieffer00code}, which matches 
the information-theoretic lower bound on the size of a grammar for a string of length~$n$.
Comparing this size with the size of the smallest grammar, its approximation ratio has
\Oh{(n / \lg n)^{2/3}} as an upper bound~\cite{charikar05grammar} and \Om{\lg n / \lg \lg n} as a lower bound~\cite{2019arXiv190806428B}.

On the practical side, \citet{yoshida13repair} presented an efficient fixed-length code for compressing the Re-Pair grammar.
Although conceived as a grammar for compressing texts, 
Re-Pair has been successfully applied for compressing
trees~\cite{lohrey13repair}, 
matrices~\cite{tabei16repair},
or images~\cite{2019arXiv190110744D}.

For different settings or for better compression rates, there is a great interest in modifications to Re-Pair.
\citet[Sect.~G]{charikar05grammar} give an easy variation to improve the size of the grammar.
\citet{sekine14repair} provide an adaptive variant whose algorithm divides the input into blocks, 
and processes each block based on the rules obtained from the grammars of its preceding blocks.
Subsequently, \citet{masaki16repair} gave an \emph{online} algorithm producing a grammar mimicking Re-Pair.
\citet{ganczorz17repair} modified the Re-Pair grammar by disfavoring the replacement of bigrams that cross Lempel-Ziv-77 (LZ77)~\cite{ziv77lz} factorization borders, 
which allowed the authors to achieve practically smaller grammar sizes.
Recently, \citet{furuya19repair} presented a variant, called \emph{MR-Re-Pair}, in which a most frequent maximal repeat is replaced instead of a most frequent bigram.

\subsection{Related Work}
Although Re-Pair is a well received grammar, there is not much literature found on how to compute Re-Pair efficiently.
\J{%
In this article, we focus on the problem to compute the grammar with an algorithm working in text space, 
forming a bridge between the domain of in-place string algorithms and the domain of Re-Pair computing algorithms.
We briefly review some prominent achievements in both domains:

\block{In-Place String Algorithms}
For the LZ77 factorization,
\citet{karkkainen13lz77} present an algorithm computing this factorization with \Oh{n/d} words on top of the input space in \Oh{dn} time for a variable $d \ge 1$, achieving \Oh{1} words with \Oh{n^2} time.
For the suffix sorting problem,
\citet{2017arXiv170301009G} gave an algorithm to compute the suffix array with \Oh{\lg n} bits on top of the output in \Oh{n} time if each character of the alphabet is present in the text. 
This condition got improved to alphabet sizes of at most~$n$ by \citet{li18sa}.
Finally, \citet{crochemore15bwt} showed how to transform a text into its Burrows-Wheeler transform by using \Oh{\lg n} of additional bits.
Due to \citet{louza17bwt}, this algorithm got extended to compute simultaneously the LCP array with \Oh{\lg n} bits of additional working space.

\block{Re-Pair Computation}
}%
Re-Pair is a grammar proposed by~\citet{larsson99repair}, who gave an algorithm computing it in expected linear time with $5n + 4\sigma^2 + 4\sigma' + \sqrt{n}$ words of working space, where $\sigma'$ is the number of non-terminals (produced by Re-Pair).
This space requirement got improved by \citet{bille17repair}, who presented a linear time algorithm taking
$(1+\epsilon)n + \sqrt{n}$ words 
on top of the rewriteable text space
for a constant~$\epsilon$ with $0 < \epsilon \le 1$.
Subsequently, they improved their algorithm in~\cite{2017arXiv170408558B} to include the text space within the $(1+\epsilon)n + \sqrt{n}$ words of working space.
However, they assume that the alphabet size~$\sigma$ is constant and $\upgauss{\lg \sigma} \le w/2$, where $w$ is the machine word size.
They also provide a solution for $\epsilon = 0$ running in expected linear time.
Recently, \citet{sakai19repair} showed how to convert an arbitrary grammar (representing a text) into the Re-Pair grammar in compressed space, i.e., without decompressing the text. Combined with a grammar compression that can process the text in compressed space in a streaming fashion, this result leads to the first Re-Pair computation in compressed space.

\block{Our Contribution}
In this article, we propose an algorithm that computes the Re-Pair grammar
in \JO[\Oh{n^2}]{$\Oh{n^2} \cap \Oh{n \lg \log_{\tau} n \lg \lg \lg n / \log_{\tau} n}$} time \JO[(cf.\ \cref{thmGoal})]{(cf.\ \cref{thmGoal} and \cref{thmWordPack})}
with $\max( (n/c) \lg n, n \upgauss{\lg \tau}) + \Oh{\lg n}$ bits of working space including the text space, where $\tau$ is the number of terminals and non-terminals.
Given that the characters of the text are drawn from a large integer alphabet with size $\sigma = \Om{n}$,
the algorithm works in-place.
This is the first non-trivial in-place algorithm, as a trivial approach on a text~$T$ of length~$n$ would compute the most frequent bigram in \Ot{n^2} time by
computing the frequency of each bigram~$T[i]T[i+1]$ for every integer~$i$ with $1 \le i \le n-1$, keeping only the most frequent bigram in memory.
This sums up to \Oh{n^3} total time, 
\dknote*{An example for this?}{and can be \Ot{n^3} for some texts} 
since there can be \Ot{n} different bigrams considered for replacement by Re-Pair.
To achieve our goal of \Oh{n^2} total time, we first provide a trade-off algorithm (cf. \cref{lemBatchedCount}) finding the $d$ most frequent bigrams in \Oh{n^2 \lg d / d} time for a trade-off parameter~$d$.
We subsequently run this algorithm for increasing values of $d$, 
and show that we need to run it $\Oh{\lg n}$ times, which gives us \Oh{n^2} time if $d$ is increasing sufficiently fast.
\J[%
Our major tools are appropriate text partitioning, elementary scans, and sorting steps, which we visualize in \cref{secStepTurn} by an example.
]{%
Our major tools are appropriate text partitioning, elementary scans, and sorting steps, which we visualize in \cref{secStepTurn} by an example, and practically evaluate in \cref{secImplementation}.
When $\tau = \oh{n}$, a different approach using word-packing and bit-parallel techniques becomes attractive,
leading to an \Oh{n \lg \log_{\tau} n \lg \lg \lg n / \log_{\tau} n} time algorithm, which we explain in \cref{secWordPack}.
  Our algorithm can be parallelized (\cref{secParallel}), used in external memory (\cref{secEM}), or adapted
to compute the MR-Re-Pair grammar in small space (\cref{secMaxRepeat}).
Finally, in \cref{secHeuristics} we study several heuristics that make the algorithm faster on specific texts.
}%

\subsection{Preliminaries}

We use the word RAM model with a word size of $\Om{\lg n}$ for an integer~$n \ge 1$.
\J{We work in the restore model~\cite{chan18restore}, in which
algorithms are allowed to overwrite the input, as long as they can
restore the input to its original form.}%

\block{Strings}
Let $T$ be a text of length~$n$ whose characters are drawn from an integer alphabet~$\Sigma$ of size $\sigma = n^{\Oh{1}}$.
A bigram is an element of $\Sigma^2$.
The \teigi{frequency} of a bigram~$B$ in $T$ is the number of \emph{non-overlapping} occurrences of~$B$ in $T$,
which is at most $\abs{T}/2$.

\block{Re-Pair}
We reformulate the recursive description in the introduction by dividing a Re-Pair construction algorithm into turns. 
Stipulating that $T_i$ is the text after the $i$-th turn with $i \ge 1$ and $T_0 := T \in \Sigma_0^+$ with $\Sigma_0 := \Sigma$,
Re-Pair replaces one of the most frequent bigrams (ties are broken arbitrarily) in $T_{i-1}$ with a non-terminal in the $i$-th turn.
Given this bigram is $\texttt{bc} \in \Sigma^2_{i-1}$,
Re-Pair replaces all occurrences of \texttt{bc} with a new non-terminal~$X_{i}$ in $T_{i-1}$, and
sets $\Sigma_{i} := \Sigma_{i-1} \cup \{X_{i}\}$ with $\sigma_i := |\Sigma_i|$ to produce $T_{i} \in \Sigma_i^+$.
Since $\abs{T_i} \le \abs{T_{i-1}}-2$, Re-Pair terminates after $m < n/2$ turns such that $T_m \in \Sigma_m^+$ contains no bigram occurring more than once.

\JO[\section{Algorithm}]{\section{Sequential Algorithm}} \label{secSequential}
A major task for producing the Re-Pair grammar is to count the frequencies of the most frequent bigrams.
Our work horse for this task are frequency tables. 
A \teigi{frequency table} in $T_i$ of length~$f$ stores pairs of the form $(\texttt{bc},x)$, where \texttt{bc} is a bigram and $x$ the frequency of \texttt{bc} in $T_i$.
It uses $f \upgauss{\lg (\sigma_i^2 n_i/2)}$ bits of space since an entry stores a bigram consisting of two characters from $\Sigma_i$ and its respective frequency, which can be at most $n_i/2$. %
Throughout this paper, we use an elementary in-place sorting algorithm like heapsort:
\begin{lemma}[\cite{william64heapsort}]
An array of length~$n$ can be sorted in-place in \Oh{n \lg n} time.
\end{lemma}

\JO[\vspace{-1em}]{}
\subsection{Trade-Off Computation}
By embracing the frequency tables, we present a solution with a trade-off parameter:

\begin{lemma}\label{lemBatchedCount}
   Given an integer~$d$ with $d \ge 1$,
   we can compute the frequencies of the $d$ most frequent bigrams in a text of length $n$ whose characters are drawn from an alphabet of size~$\sigma$ 
   in \Oh{\max(n,d) n  \lg d / d} time using $2d \upgauss{\lg(\sigma^2 n/2)} + \Oh{\lg n}$ bits.
\end{lemma}
\begin{proof}
   Our idea is to partition the set of all bigrams appearing in $T$ into $\upgauss{n / d}$ subsets, 
   compute the frequencies for each subset, and finally merge these frequencies.
   In detail, we partition the text~$T = S_1 \cdots S_{\upgauss{n / d}}$ into $\upgauss{n / d}$ substrings such that each substring has length~$d$ (the last one has a length of at most $d$).
   Subsequently, we extend $S_j$ to the left (only if $j > 1$) and to the right (only if $j < \upgauss{n / d}$) such that
$S_j$ and $S_{j+1}$ overlap by one text position, for $1 \le j < \upgauss{n / d}$.
By doing so, we take the bigram on the border of two adjacent substrings~$S_j$ and $S_{j+1}$ for each~$j < \upgauss{n / d}$ into account.
   Next, we create two frequency tables~$F$ and~$F'$, each of length $d$ for storing the frequencies of $d$ bigrams.
With $F$ and $F'$, we process each of the $n/d$ substrings~$S_j$ as follows:
Let us fix an integer~$j$ with $1 \le j \le \upgauss{n/d}$. 
We first put all bigrams of $S_j$ into $F'$ in \emph{lexicographic} order.
We can perform this within the space of $F'$ in \Oh{d \lg d} time since there are at most $d$ different bigrams in $S_j$.
We compute the frequencies of all these bigrams in the \emph{complete} text~$T$ in $\Oh{n \lg d}$ time
by scanning the text from left to right while locating a bigram in $F'$ in \Oh{\lg d} time with a binary search.
Subsequently, we interpret $F$ and $F'$ as one large frequency table, sort it with respect to the frequencies while discarding duplicates
such that $F$ stores the $d$ most frequent bigrams in~$T[1..jd]$.
This sorting step can be done in \Oh{d \lg d} time.
Finally, we clear $F'$ and are done with $S_j$.
After the final merge step, we obtain the $d$ most frequent bigrams of~$T$ stored in~$F$.

Since each of the \Oh{n/d} merge steps takes \Oh{d \lg d + n \lg d} time, we need
\(
\Oh{\max(d,n) \cdot (n \lg d)/d}
   \)
   time.
   For $d \ge n$, we can build a large frequency table and perform one scan to count the frequencies of all bigrams in $T$.
   This scan and the final sorting with respect to the counted frequencies can be done in \Oh{n \lg n} time.
\end{proof}

\subsection{Algorithmic Ideas}

With \cref{lemBatchedCount}, we can compute $T_m$ in \Oh{m n^2 \lg d / d} time with additional $2d$ $\upgauss{\lg(\sigma_m^2 n/2)}$~bits of working space on top of the text for a parameter~$d$ with $1 \le d \le n$.
In the following, we present an \Oh{n^2} time algorithm that needs $\max( (n/c) \lg n,$ $n \upgauss{\lg \sigma_m}) + \Oh{\lg n}$ bits of working space, where the text space is included as a rewriteable part in the working space and $c \ge 1$ is a constant.
In this model, we assume that we can enlarge the text~$T_i$ from $n_i \upgauss{\lg \sigma_i}$ bits to $n_i \upgauss{\lg \sigma_{i+1}}$ bits without additional extra memory.
Our main idea is to store a growing frequency table using the space freed up by replacing bigrams with non-terminals.
In detail, we maintain a frequency table $F$ in $T_i$ of length $f_k$ for a growing variable~$f_k$,
which is set to $f_0 := \Oh{1}$ in the beginning.
The table~$F$ takes $f_k \upgauss{\lg (\sigma_i^2 n/2)}$ bits, which is $\Oh{\lg (\sigma^2 n)} = \Oh{\lg n}$ bits for $k = 0$.
When we want to query it for a most frequent bigram, we linearly scan~$F$ in $\Oh{f_k} = \Oh{n}$ time,
which is not a problem since (a) the number of queries is $m \le n$, and (b) we aim for \Oh{n^2} overall running time.
A consequence is that there is no need to sort the bigrams in $F$ according to their frequencies, which simplifies the following discussion.

\paragraph{Frequency Table~$F$.}
With \cref{lemBatchedCount}, we can compute $F$ in $\bigOh(n \max(n, f_k)$ $\lg f_k / f_k)$ time.
Instead of recomputing $F$ for every turn~$i$, we want to recompute it only when it no longer stores a most frequent bigram.
However, it is ad-hoc not clear when this happens as replacing a most frequent bigram during a turn (a) removes this entry in~$F$
and (b) can reduce the frequencies of other bigrams in~$F$, making them possibly less frequent than other bigrams not tracked by~$F$.
Hence, the variable~$i$ for the $i$-th turn (creating the $i$-th non-terminal) and the variable~$k$ for recomputing the frequency table~$F$ the $(k+1)$-st time are loosely connected. 
We group together all turns with the same $f_k$ and call this group the \teigi{$k$-th round} of the algorithm.
At the beginning of each round, we enlarge~$f_k$ and create a new~$F$ with a capacity for $f_k$~bigrams.
Since a recomputation of~$F$ takes much time,
we want to end a round only if~$F$ is no longer useful, i.e., when we no longer can guarantee that $F$ stores a most frequent bigram.
To achieve our claimed time bounds, we want to assign all $m$ turns to $\Oh{\lg n}$ different rounds,
which can only be done if $f_k$ grows sufficiently fast.

\paragraph{Algorithm Outline.}
Given we are at the beginning of the $k$-th round and the $i$-th turn, 
we compute the frequency table $F$ storing~$f_k$ bigrams, and keep additionally the lowest frequency of~$F$ as a threshold~$t$, which is treated as a constant during this round.
During the computation of the $i$-th turn, we replace the most frequent bigram (say, $\texttt{bc} \in \Sigma_i^2$) in the text~$T_i$ with a non-terminal~$X_{i+1}$ to produce~$T_{i+1}$.
Thereafter, we remove \texttt{bc} from $F$ and update those frequencies in~$F$ which got decreased by the replacement of \texttt{bc} with $X_{i+1}$, and
add each bigram containing the new character~$X_{i+1}$ into~$F$ if its frequency is at least~$t$.
Whenever a frequency in~$F$ drops below~$t$, we discard it. 
If $F$ becomes empty, we move to the $(k+1)$-st round, and create a new $F$ for storing $f_{k+1}$ frequencies.
Otherwise ($F$ still stores an entry), we can be sure that~$F$ stores a most frequent bigram.
In both cases, we recurse with the $(i+1)$-st turn by selecting the bigram with the highest frequency stored in~$F$.
We describe in the following how we update of~$F$ and how large $f_{k+1}$ can be at least.

\subsection{Algorithmic Details}
Suppose that we are in the $k$-th round and in the $i$-th turn.
Let~$t$ be the lowest frequency in~$F$ computed at the beginning of the $k$-th round.
We keep~$t$ as a constant threshold for the invariant that all frequencies in~$F$ are at least~$t$ during the $k$-th round.
With this threshold we can assure in the following that $F$ is either empty or stores a most frequent bigram.\JO{\\}
Now suppose that the most frequent bigram of~$T_i$ is $\texttt{bc} \in \Sigma_i^2$, which is stored in $F$.
To produce $T_{i+1}$ (and hence advancing to the $(i+1)$-st turn), 
we enlarge the space of $T_i$ from $n_i \upgauss{\lg \sigma_i}$ to $n_i \upgauss {\lg \sigma_{i+1}}$, and replace all occurrences of \texttt{bc} in $T_i$ with a new non-terminal $X_{i+1}$.
Subsequently, we would like to take the next bigram of~$F$.
For that, however, we need to update the stored frequencies in $F$.
To see this necessity, suppose that there is an occurrence of \texttt{abcd} with two characters $\texttt{a}, \texttt{d} \in \Sigma_i$ in $T_i$. 
By replacing $\texttt{bc}$ with $X_{i+1}$,
\begin{enumerate}[label=(\alph*)]
  \item the frequencies of \texttt{ab} and \texttt{cd} decrease by one\footnote{For the border case \texttt{a} = \texttt{b} = \texttt{c} (resp.\ \texttt{b} = \texttt{c} = \texttt{d}), there is no need to decrement the frequency of \texttt{ab} (resp.\ \texttt{cd}).}, and
  \item the frequencies of $\texttt{a}X_{i+1}$ and $X_{i+1}\texttt{d}$ increase by one.
\end{enumerate}
\paragraph{Updating~$F$}
We can take care of the former changes~(a) by decreasing the respective bigram in $F$ (in case that it is present).
If the frequency of this bigram drops below the threshold~$t$,
we remove it from~$F$ as there may be bigrams with a higher frequency that are not present in $F$.
To cope with the latter changes~(b), we track the characters adjacent to $X_{i+1}$ after the replacement, count their numbers, and add 
their respective bigrams to $F$ if their frequencies are sufficiently high.
In detail, suppose that we have substituted \texttt{bc} with $X_{i+1}$ exactly $h$ times.
Consequently, with the new text~$T_{i+1}$ we have additionally $h \lg \sigma_{i+1}$ bits of free space\footnote{The free space is consecutive after shifting all characters to the left.}, which we call~$D$ in the following.
Subsequently, we scan the text and put the characters of $\Sigma_{i+1}$ appearing to the left of each of the $h$ occurrences of $X_{i+1}$ into~$D$.
After sorting the characters in $D$ lexicographically, we can count the frequency of $\texttt{a}X_{i+1}$ for each character $\texttt{a} \in \Sigma_{i+1}$ preceding an occurrence of $X_{i+1}$ in the text~$T_{i+1}$ by scanning~$D$ linearly.
If the obtained frequency of such a bigram $\texttt{a}X_{i+1}$ is at least as high as the threshold~$t$,
we insert $\texttt{a}X_{i+1}$ into $F$,
and subsequently discard a bigram with the currently lowest frequency in~$F$ if the size of $F$ has become $f_k+1$.
In case that we visit a run of $X_{i+1}$'s during the creation of~$D$, we must take care of not counting the overlapping occurrences of $X_{i+1} X_{i+1}$.
Finally, we can count analogously the occurrences of $X_{i+1}\texttt{d}$ for all characters $\texttt{d} \in \Sigma_i$ succeeding an occurrence of $X_{i+1}$.

\paragraph{Capacity of~$F$}
After the above procedure we have updated the frequencies of $F$.
When $F$ becomes empty, we end the $k$-th round and continue with the ($k+1$)-st round by creating a new frequency table~$F$ with capacity~$f_{k+1}$.
In what follows, we (a) analyze in detail when $F$ becomes empty (as this determines the sizes~$f_k$ and~$f_{k+1}$),
and (b) show that we can compensate the number of discarded bigrams with an enlargement of $F$'s capacity from $f_k$~bigrams to $f_{k+1}$~bigrams
for the sake of our aimed total running time:
If the frequency of \texttt{bc} in $T_i$ is $x$, then we can reduce at most $2x$ frequencies of other bigrams.
Since a bigram must occur at least twice in~$T_i$ to be present in~$F$,
the frequency of \texttt{bc} has to be at least $\max(2, (f_k-1)/2)$ for discarding all bigrams of~$F$, and
each replacement of \texttt{bc} with $X_{i+1}$ frees up $\upgauss{\lg\sigma_{i+1}}$ bits of the text.

Suppose that we have enough space available for storing the frequencies of $\alpha f_k$ bigrams, 
where $\alpha$ is a constant (depending on $\sigma_i$ and $n_i$) such that $F$ and the working space of \cref{lemBatchedCount} with $d = f_k$ can be stored within this space.
Let $\delta := \lg (\sigma^2_{i+1} n_i/2)$ be the number of bits needed to store one entry in $F$,
and let $\beta := \min(\delta / \lg \sigma_{i+1}, c \delta/\lg n)$ 
be the minimum number of characters that need to be freed to store one frequency in this space.
To understand the value of $\beta$, we look at the arguments of the minimum function in the definition of~$\beta$ and simultaneously at the maximum function 
in our aimed working space of $\max(n \upgauss{\lg \sigma_m},  (n/c) \lg n) + \Oh{\lg n}$ bits (cf.~\cref{thmGoal}):
\begin{itemize}
  \item The first item in this maximum function allows us to spend $\lg \sigma_{i+1}$ bits for each freed character such that we 
    obtain space for one additional entry in~$F$ after freeing $\delta / \lg \sigma_{i+1}$ characters.
  \item The second item allows us to use $\lg n$ additional bits after freeing up $c$ characters.%
\footnote{This additional treatment helps us to let $f_k$ grow sufficiently fast in the first steps to save our \Oh{n^2} time bound, as for sufficiently small alphabets and large text sizes, $\lg (\sigma^2 n/2) / \lg \sigma = \Oh{\lg n}$, which means that we might run the first $\Oh{\lg n}$ turns with $f_k = \Oh{1}$, and therefore already spend \Oh{n^2 \lg n} time.}
    Hence, after freeing up $c \delta/\lg n$ characters, we have space to store one additional entry in~$F$.
\end{itemize}
\begin{align*}
  \text{With $\beta$, we have~}
   \alpha f_{k+1} 
   &= \alpha f_k + \max(2/\beta, (f_k-1)/(2\beta)) \\
   &= \alpha f_k \max(1  + 2/(\alpha \beta f_k) , 1 + 1/(2\alpha\beta) - 1/(2 \alpha\beta f_k))
 \\&\ge \alpha f_k (1 + 2/(5\alpha\beta)) =: \gamma_i \alpha f_k \text{~with~} \gamma_i := 1 + 2/(5\alpha\beta),
   \end{align*}
where we used the equivalence $1  + 2/(\alpha \beta f_k) = 1 + 1/(2\alpha\beta) - 1/(2 \alpha\beta f_k) \Leftrightarrow 5 = f_k$
to estimate the two arguments of the maximum function.

Since we let $f_k$ grow by a factor of at least $\gamma := \min_{1\le i \le n} \gamma_i > 1$ for each recomputation of $F$, 
$f_k = \Om{\gamma^k}$, and therefore $f_{k} = \Ot{n}$ after $k = \Oh{\lg n}$ steps.
Consequently, after reaching $k = \Oh{\lg n}$, we can iterate the above procedure a constant number of times to compute the non-terminals of the remaining bigrams occurring at least twice.

\paragraph{Time Analysis}
On the total picture, we compute $F$ \Oh{\lg n} times with \cref{lemBatchedCount}.
For the $k$-th time, we run the algorithm of \cref{lemBatchedCount} with $d = f_k$ 
on a text of length at most $n-f_k$ in \Oh{n(n-f_k) \cdot \lg f_k / f_k} time with $f_k \le n$.
Summing this up, we yield
\begin{equation}\label{eqTotalTime}
   \OhS{\sum_{k=0}^{\Oh{\lg n}} \frac{n-f_k}{f_k} n \lg f_k} = \OhS{n^2 \sum_k^{\lg n} \frac{k}{\gamma^k}} = \OhS{n^2} \textup{~time in total.}
\end{equation}
In the $i$-th turn, we update $F$ by decreasing the frequencies of the bigrams affected by the substitution of the most frequent bigram \texttt{bc} with $X_{i}$.
For decreasing such a frequency, we look up its respective bigram with a linear scan in $F$, which takes $f_k = \Oh{n}$ time.
However, since this decrease is accompanied with a replacement of an occurrence of \texttt{bc}, 
we obtain \Oh{n^2} total time by charging each text position with $\Oh{n}$ time for a linear search in $F$.
With the same argument, we can bound the total time for sorting the characters in $D$ to \Oh{n^2} overall time:
Since we spend $\Oh{h \lg h}$ time on sorting $h$ characters preceding or succeeding a replaced character, 
and $\Oh{f_k} = \Oh{n}$ time on swapping a sufficiently large new bigram composed of $X_{i+1}$ and a character of $\Sigma_{i+1}$ with
a bigram with the lowest frequency in~$F$,
we charge each text position again with \Oh{n} time.
Putting all time bounds together leads to the main result of this article:

\begin{theorem}\label{thmGoal}
   We can compute Re-Pair on a string of length~$n$ in \Oh{n^2} time 
   with $\max( (n/c) \lg n, n \upgauss{\lg \sigma_m}) + \Oh{\lg n}$ bits of working space including the text space,
   where $c \ge 1$ is a fixed constant, and $\sigma_m$ is the number of terminal and non-terminal symbols.
\end{theorem}

\paragraph{Output}
Finally, we show that we can store the computed grammar in text space.
More precisely, we want to store the grammar in an auxiliary array~$A$ packed at the end of the working space such that
the entry $A[i]$ stores the right hand side of the non-terminal~$X_i$, which is a bigram.
Thus the non-terminals are represented implicitly as indices of the array~$A$.
We therefore need to subtract $2\lg \sigma_i$ bits of space from our working space~$\alpha f_k$ after the $i$-th turn.
By adjusting~$\alpha$ in the above equations, we can deal with this additional space requirement
as long as the frequencies of the replaced bigrams are at least three (we charge two occurrences for growing the space of $A$).

When only bigrams with frequencies of at most two remain, we switch to a simpler algorithm, 
discarding the idea of maintaining the frequency table~$F$:
Suppose that we work with the text~$T_i$.
Let $k$ be a text position, which is $1$ in the beginning, but will be incremented in the following turns while
holding the invariant that $T[1..k]$ does not contain a bigram of frequency two.
We scan $T_i[k..n]$ linearly from left to right and check, for each text position~$j$, 
whether the bigram~$T_i[j]T_i[j+1]$ has another occurrence $T_i[j']T_i[j'+1] = T_i[j]T_i[j+1]$ with $j' > j+1$, and if so, 
\begin{enumerate}[label=(\alph*)]
  \item   append $T_i[j]T_i[j+1]$ to~$A$,
  \item   replace $T_i[j]T_i[j+1]$ and $T_i[j']T_i[j'+1]$ with a new non-terminal~$X_{i+1}$ to transform $T_i$ to $T_{i+1}$, and
  \item   recurse on $T_{i+1}$ with $k := j$ until no bigram with frequency two is left.
\end{enumerate}
The position~$k$, which we never decrement, helps us to skip over all text positions starting with bigrams with a frequency of one.
Thus, the algorithm spends \Oh{n} time for each such text position, and \Oh{n} time for each bigram with frequency two.
Since there are at most $n$ such bigrams, the overall running time of this algorithm is \Oh{n^2}.

\J{%
\begin{remark}[Pointer Machine Model]
Refraining from the usage of complicated algorithms,
our algorithm consists only of elementary sorting and scanning steps.
This allows us to run our algorithm on a pointer machine, yielding the same time bound of \Oh{n^2}.
For the space bounds, we assume that the text is given in $n$ words, 
where a word is large enough to store an element of $\Sigma_m$ or a text position.
\end{remark}
}%

\begin{figure}[ht]
\begin{center}
  \includegraphics[width=\linewidth]{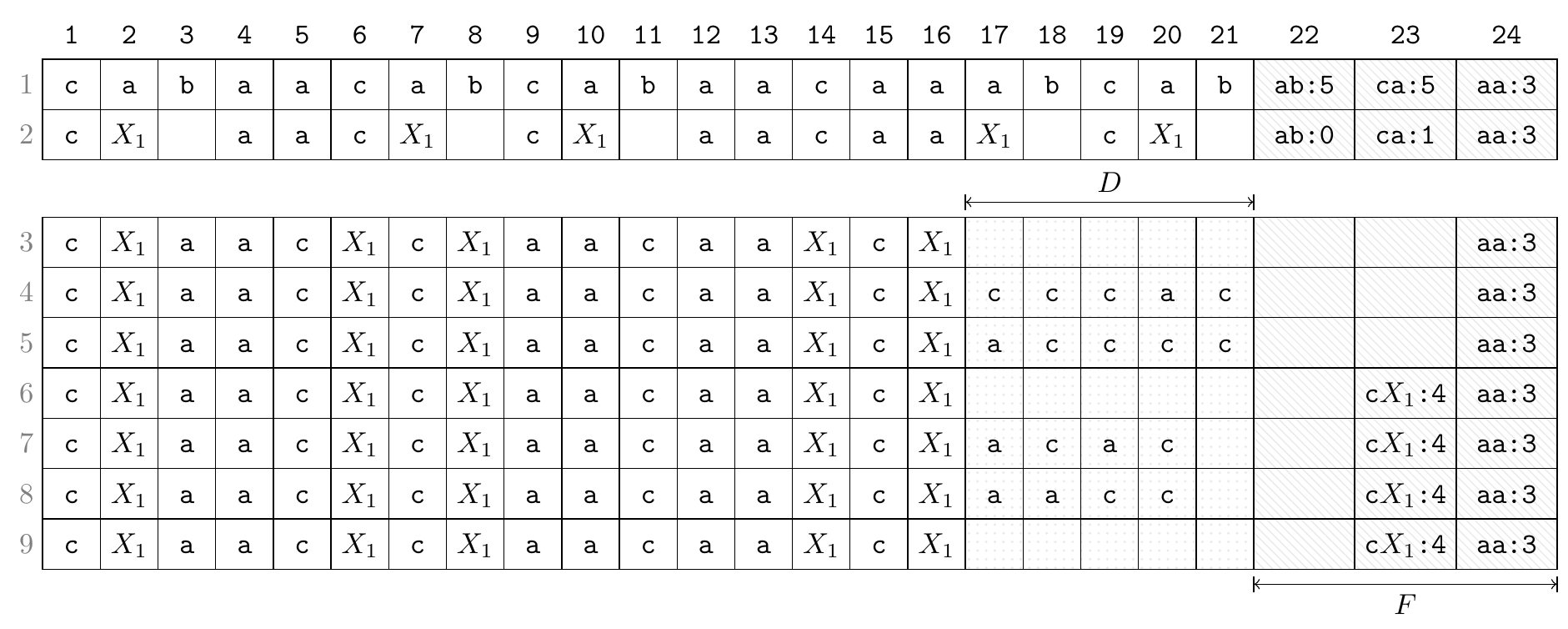}
  \caption{Step-by-step execution of the first turn of our algorithm on the string $T = \texttt{cabaacabcabaacaaabcab}$.
   The turn starts with the memory configuration given in Row~1. 
    Positions~1 to~21 are text positions, positions~22 to~24 belong to~$F$ ($f_0 = 3$, and it is assumed that a frequency fits into a text entry).
   Subsequent rows depict the memory configuration during Turn~1.
    A comment to each row is given in \cref{secStepTurn}. 
  }
  \label{figFirstTurn}
\end{center}
\end{figure}

\subsection{Step-by-Step Execution} \label{secStepTurn}
Here, we present an exemplary execution of the first turn (of the first round) on the input $T = \texttt{cabaacabcabaacaaabcab}$.
We visualize each step of this turn as a row in \cref{figFirstTurn}. 
A detailed description of each row follows:
\begin{description}
  \item[Row~1:] Suppose that we have computed $F$, which has a constant number of entries\footnote{In the later turns when the size $f_k$ becomes larger, $F$ will be put in the text space.}.
    The highest frequency is five achieved by $\mathtt{ab}$ and $\mathtt{ca}$.
               The lowest frequency represented in $F$ is three, 
               which becomes the threshold for a bigram to be present in $F$ such that bigrams whose frequencies drop below this threshold are removed from~$F$.
               This threshold is a constant for all later turns until $F$ is rebuilt (in the following round).
               During Turn~1, the algorithm proceeds now as follows:
  \item[Row~2:] Choose $\mathtt{ab}$ as a bigram to replace with a new non-terminal~$X_1$ (break ties arbitrarily).
               Replace every occurrence of $\mathtt{ab}$ with $X_1$ while decrementing frequencies in $F$ accordingly to the neighboring characters of the replaced occurrence.
  \item[Row~3:] Remove from $F$ every bigram whose frequency falls below the threshold.
               Obtain space for $D$ by aligning the compressed text~$T_1$.
               (The process of Row~2 and Row~3 can be done simultaneously.)
             \item[Row~4:] Scan the text and copy each character preceding an occurrence of~$X_1$ in~$T_1$ to~$D$.
             \item[Row~5:] Sort characters in $D$ lexicographically.
             \item[Row~6:] Insert new bigrams (consisting of a character of~$D$ and~$X_1$) whose frequencies are at least as large as the threshold.
             \item[Row~7:] Scan the text again and copy each character succeeding an occurrence of~$X_1$ in~$T_1$ to~$D$ (symmetric to Row~4).
             \item[Row~8:] Sort all characters in $D$ lexicographically (symmetric to Row~5).
             \item[Row~9:] Insert new bigrams whose frequencies are at least as large as the threshold (symmetric to Row~6).
\end{description}

\JO[%
\paragraph*{Acknowledgments}
\Acknowledgments{}
]{}

\subsection{Implementation}\label{secImplementation}
We provide a simplified implementation in C++17 at \url{https://github.com/koeppl/repair-inplace}.
The simplification is that we (a) fix the bit width of the text space to 16 bit, and (b) assume that $\Sigma$ is the byte alphabet.
We further skip the step increasing the bit width from $\lg \sigma_{i}$ to $\lg \sigma_{i+1}$.
This means that the program works as long as the characters of~$\Sigma_m$ fit into 16 bits.
The benchmark, whose results are displayed in \cref{tableEvaluation}, was conducted on a Mac Pro Server with an Intel Xeon CPU X5670 clocked at 2.93GHz running Arch Linux.
The implementation was compiled with \texttt{gcc-8.2.1} in the highest optimization mode \texttt{-O3}.
Looking at \cref{tableEvaluation}, we can see that the running time is super-linear to the input size on all text instances, 
which we obtained from the Pizza\&Chili corpus (\url{http://pizzachili.dcc.uchile.cl/}).
\Cref{tableDatasets} gives some characteristics about the used data sets.
We see that the number of rounds is the number of turns plus one for every unary string $\texttt{a}^{2^k}$ with an integer $k \ge 1$
since the text contains only one bigram with a frequency larger than two in each round. 
Replacing this bigram in the text makes~$F$ empty such that the algorithm recomputes~$F$ after each turn.
Note that the number of rounds can drop while scaling the prefix length based on the choice of the bigrams stored in~$F$.

\begin{table}
  \setlength{\tabcolsep}{0.6em}
  \centerline{%
    \begin{tabular}{l@{\hskip 1em}*{5}{r}}
\toprule
 & \multicolumn{5}{c}{Prefix Size in KiB} \\
\cmidrule{2-6}
Data Set  &  64 & 128 & 256 & 512 & 1024 \\
\midrule
\textsc{Escherichia\_Coli} & 20.68 & 130.47 & 516.67 & 1708.02 & 10112.47
\\
\textsc{cere}              & 13.69 & 90.83  & 443.17 & 2125.17 & 9185.58
\\
\textsc{coreutils}         & 12.88 & 75.64  & 325.51 & 1502.89 & 5144.18
\\
\textsc{einstein.de.txt}   & 19.55 & 88.34  & 181.84 & 805.81  & 4559.79
\\
  \textsc{einstein.en.txt} & 21.11 & 78.57  & 160.41 & 900.79  & 4353.81
\\
    \textsc{influenza}     & 41.01 & 160.68 & 667.58 & 2630.65 & 10526.23
\\
      \textsc{kernel}      & 20.53 & 101.84 & 208.08 & 1575.48 & 5067.80
\\
        \textsc{para}      & 20.90 & 175.93 & 370.72 & 2826.76 & 9462.74
\\
\textsc{world\_leaders}    & 11.92 & 21.82  & 167.52 & 661.52  & 1718.36
\\
\cdashlinelr{1-6}
$\texttt{aa}\cdots\texttt{a}$    & 0.35 & 0.92 & 3.90 & 14.16 & 61.74
\\
\bottomrule
\end{tabular}
  }%
  \caption{Experimental evaluation of our implementation described in \cref{secImplementation}. Table entries are running times in seconds.
    The last line is the benchmark on the unary string $\texttt{aa}\cdots\texttt{a}$.
  }
  \label{tableEvaluation}
\end{table}

\begin{table}
  \setlength{\tabcolsep}{0.6em}
  \centerline{%
    \begin{tabular}{l@{\hskip 1em}r@{\hskip 1em}*{5}{r}@{\hskip 1em}*{5}{r}}
\toprule
 &   &  \multicolumn{5}{c}{Turns $/1000$} &  \multicolumn{5}{c}{Rounds} \\
& & \multicolumn{5}{c}{Prefix Size in KiB} & \multicolumn{5}{c}{Prefix Size in KiB} \\
\cmidrule(lr){3-7} 
\cmidrule(lr){8-12} 
Data Set &  $\sigma$ 
 & $2^6$ & $2^7$ & $2^8$ & $2^9$ & $2^{10}$ 
 & $2^6$ & $2^7$ & $2^8$ & $2^9$ & $2^{10}$ \\
\midrule
\textsc{Escherichia\_Coli}  & 4 & 1.8 & 3.2 & 5.6  & 10.3 & 18.1 & 6  & 9  & 9  & 12 & 12
\\
\textsc{cere}               & 5 & 1.4 & 2.8 & 5.0  & 9.2  & 15.1 & 13 & 14 & 14 & 14 & 14
\\
\textsc{coreutils}          & 113 & 4.7 & 6.7 & 10.2 & 16.1 & 26.5 & 15 & 15 & 15 & 14 & 14
\\
\textsc{einstein.de.txt}    & 95 & 1.7 & 2.8 & 3.7  & 5.2  & 9.7  & 14 & 14 & 15 & 16 & 16
\\
  \textsc{einstein.en.txt}  & 87 & 3.3 & 3.5 & 3.8  & 4.5  & 8.6  & 16 & 15 & 15 & 15 & 17
\\
    \textsc{influenza}      & 7 & 2.5 & 3.7 & 9.5  & 13.4 & 22.1 & 11 & 12 & 14 & 13 & 15
\\
      \textsc{kernel}       & 160 & 4.5 & 8.0 & 13.9 & 24.5 & 43.7 & 10 & 11 & 14 & 14 & 13
\\
        \textsc{para}       & 5 & 1.8 & 3.2 & 5.8  & 10.1 & 17.6 & 12 & 12 & 13 & 13 & 14
\\
\textsc{world\_leaders}     & 87 & 2.6 & 4.3 & 6.1  & 10.0 & 42.1 & 11 & 11 & 11 & 11 & 14
\\
\cdashlinelr{1-12}
$\texttt{aa}\cdots\texttt{a}$ & 1 & 15 & 16 & 17 & 18 & 19 & 16 & 17 & 18 & 19 & 20
\\
\bottomrule
\end{tabular}
  }%
  \caption{Characteristics of our data sets. The number of turns and rounds are given for each of the prefix sizes 128, 256, 512, and 1024 KiB of the respective data sets. The number of turns reflecting the number of non-terminals is given in units of thousands.
    The turns of the unary string $\texttt{aa}\cdots\texttt{a}$ are in plain units (not divided by thousand).
  }
  \label{tableDatasets}
\end{table}

\section{Bit-Parallel Algorithm}\label{secWordPack}
In the case that the number of terminals and non-terminals $\tau := \sigma_m$ is $\oh{n}$, a word-packing approach becomes interesting.
We present techniques speeding up previously introduced operations on chunks of $\Oh{\log_\tau n}$ characters from \Oh{\log_\tau n} time to \Oh{\lg \lg \lg n} time.
In the end, these techniques allow us to speed up the sequential algorithm of \cref{thmGoal} from \Oh{n^2} time to the following:

\begin{theorem}\label{thmWordPack}
  We can compute Re-Pair on a string of length~$n$ in $\bigOh(n^2 \lg \log_\tau n$ $\lg \lg \lg n / \log_\tau n)$ time 
   with $\max( (n/c) \lg n, n \upgauss{\lg \tau}) + \Oh{\lg n}$ bits of working space including the text space,
   where $c \ge 1$ is a fixed constant, and $\tau$ is the number of terminal and non-terminal symbols.
\end{theorem}
Note that the \Oh{\lg \lg \lg n} factor is due to the $\fnPopcount$ function~\cite[Algo.~1]{vigna08broadword}, which has been optimized to a single instruction on modern computer architectures.

\subsection{Broadword Search}\label{secBroadwordSearch}

First, we deal with accelerating the computation of the frequency of a bigram in $T$ by exploiting broadword search thanks to the word RAM model.
We start with the search of single characters and subsequently extend this result to bigrams:

\J{%
    \begin{figure}
      \centerline{%
        \begin{tabularx}{\linewidth}{l@{\hskip 1em}X}
          \toprule
          Operation & Description
          \\\midrule
          $X \ll j$ & shift $X$ $j$~positions to the left \\
          $X \gg j$ & shift $X$ $j$~positions to the right \\
          $\neg X$ & bitwise NOT of $X$ \\
          $X \otimes Y$ & bitwise XOR of $X$ and $Y$ \\
$-1$ & bit vector consisting only of one bits \\
$\fnMSB(X)$ & returns the position of the most significant set bit of~$X$, i.e., $\gauss{\lg X}+1$; see~\cite[Sect.~5]{fredman93fusion} for a constant time algorithm using {\Oh{\lg n}} bits \\
$\fnDeletePrefixRun(X)$ & sets all bits of the maximal prefix of consecutive ones to zero \\
$\fnDeleteSuffixRun(X)$ & sets all bits of the maximal suffix of consecutive ones to zero 
          \\\bottomrule
        \end{tabularx}
      }%
      \caption{Operations used in \cref{figBroadwordMatching,figBroadwordBigram} for two bit vectors~$X$ and~$Y$.
      All operations can be computed in constant time.
      See \cref{tableExampleDeletePrefixRun} for an example of $\fnDeleteSuffixRun$ and $\fnDeletePrefixRun$.
    }
      \label{tableBroadwordDesc}
    \end{figure}

\begin{figure}
  \centering
  \begin{tabular}[t]{l@{\hskip 0.5em}l}
          \multicolumn{2}{c}{$\fnDeletePrefixRun(X)$} \\
          \toprule
          Operation & Example
          \\\midrule
    $X$                                                   & \texttt{11100110}\\
    $\neg X$                                              & \texttt{00011001}\\
    $1 \ll (1+\fnMSB(\neg X))$                              & \texttt{00100000}\\
    $(1 \ll (1+\fnMSB(\neg X)))-1$                        & \texttt{00011111}\\
    $((1 \ll (1+\fnMSB(\neg X)))-1) \fnAND X $            & \texttt{00000110}
          \\\bottomrule
        \end{tabular}
        \begin{tabular}[t]{l@{\hskip 0.5em}l}
          \multicolumn{2}{c}{$\fnDeleteSuffixRun(X)$} \\
          \toprule
          Operation                                       & Example
          \\\midrule
    $X$                                                   & \texttt{01100111}\\
    $\neg X$                                              & \texttt{10011000}\\
    $\neg X - 1$                                          & \texttt{10010111}\\
    $(\neg X - 1) \fnAND X$                               & \texttt{00000111}\\
    $\neg((\neg X - 1) \fnAND X)$                         & \texttt{11111000}\\
    $\neg((\neg X - 1) \fnAND X) \fnAND X$                & \texttt{01100000}
          \\\bottomrule
        \end{tabular}

      \caption{%
        Step-by-step execution of $\fnDeletePrefixRun(X)$ and $\fnDeleteSuffixRun(X)$
        introduced in \cref{tableBroadwordDesc} on a bit vector~$X$.
      }
      \label{tableExampleDeletePrefixRun}
    \end{figure}
  }%

\begin{lemma}\label{lemBroadwordSearch}
  We can count the occurrences of a character~$c \in \Sigma$ in a string of length~$\Oh{\log_\sigma n}$ in \Oh{\lg \lg \lg n} time.
\end{lemma}
\JO[See the appendix for a proof, which is a variation of broadword searching zero bytes~{\cite[Sect.~7.1.3]{knuthArt4bitwise}}.]{}
\JProof{\cref{lemBroadwordSearch}}{%
  Let $q$ be the largest multiple of $\upgauss{\lg \sigma}$ fitting into a computer word, divided by $\upgauss{\lg \sigma}$.
Let $S \in \Sigma^*$ be a string of length~$q$.
Our first task is to compute a bit mask of length~$q \upgauss{\lg \sigma}$ marking with a \bsq{1} the occurrences of a character~$c \in \Sigma$ in~$S$.
For that, we follow the constant time broadword pattern matching of~\citet[Sect.~7.1.3]{knuthArt4bitwise}\footnote{See \url{https://github.com/koeppl/broadwordsearch} for a practical implementation.}:
Let $H$ and $L$ be two bit vectors of length $\upgauss{\lg \sigma}$ having marked only the most significant or the least significant bit, respectively.
Let $H^q$ and $L^q$ denote the $q$ times concatenation of $H$ and $L$, respectively.
Then the operations in \cref{figBroadwordMatching} yield an array~$X$ of length~$q$ with 
\begin{align}\label{eqBroadwordX}
X[i] = 
\begin{cases}
  2^{\upgauss{\lg \sigma}}-1 & \text{~if~} S[i] = c, \\
  0 & \text{~otherwise,} \\
\end{cases}
\end{align}
where each entry of $X$ has $\upgauss{\lg \sigma}$ bits.
\begin{figure}
\newcommand*{\ParboxO}[1]{}%
    \newcommand*{\ParboxC}[1]{\parbox[t]{10em}{#1}}
  \centering{%
    \begin{tabularx}{\linewidth}{l@{\hskip 1em}X@{\hskip 1em}ll}
      \toprule
      Operation & Description & & Example 
      \\\midrule
      read $S$ & & & $\texttt{101010000} \rightarrow S$
        \\
$X \gets S \fnXOR c^q$ & 
match $S$ with $c^q$; $X[i] = 0 \Leftrightarrow S[i] = c$
                       & 
                       \ParboxO{%
                         \phantom{.} \\
                         $\fnXOR$ \\
                       $=$} &
                       \ParboxC{%
                         $\texttt{101010000} = S$ \\
                       $\texttt{010010010}$ \\
                     $\texttt{111000010} \rightarrow X$}
\\ \cdashlinelr{1-4}
$Y \gets X - L^q$ &
                  & 
                  \ParboxO{%
                         \phantom{.} \\
                         $-$ \\
                       $=$} &
                       \ParboxC{%
                         $\texttt{111000010} = X$ \\
                         $\texttt{001001001}$ \\
                         $\texttt{101111001} \rightarrow Y$}
\\ \cdashlinelr{1-4}
$X \gets Y \fnAND \neg X$ &
$X[i] \fnAND 2^{\upgauss{\lg\sigma}}-1 = 1 \Leftrightarrow S[i] = c$ &
                  \ParboxO{%
                         \phantom{.} \\
                         $\fnAND$ \\
                       $=$} &
                       \ParboxC{%
                         $\texttt{101111001} = Y$ \\
                       $\texttt{000111101}$ \\
                     $\texttt{000111001} \rightarrow X$}
\\ \cdashlinelr{1-4}
$X \gets X \fnAND H^q$ &
$X[i] = 0 \Leftrightarrow S[i] \neq c$ &
                  \ParboxO{%
                         \phantom{.} \\
                         $\fnAND$ \\
                       $=$} &
                       \ParboxC{%
                         $\texttt{000111001} = X$ \\
                       $\texttt{100100100}$ \\
                     $\texttt{000100000} \rightarrow X$}
\\ \cdashlinelr{1-4}
                       $X \gets (X - (X \gg (\upgauss{\lg\sigma}-1))) \fnOR X$ &
$X$ as in \cref{eqBroadwordX} &
                  \ParboxO{%
                         \phantom{.} \\
                         $-$ \\
                         $=$ \\
                         $\fnOR$ \\
                       $=$} &
                       \ParboxC{%
                         $\texttt{000100000} = X$ \\
                       $\texttt{000001000}$ \\
                     $\texttt{000011000}$ \\
                   $\texttt{000100000}$ \\
                 $\texttt{000111000} \rightarrow X$}
      \\\bottomrule
    \end{tabularx}
  }%
  \caption{Broadword matching all occurrences of a character in a string~$S$ fitting into a computer word. For the last step, special care has to be taken when the last character of~$S$ is a match, as shifting~$X$ $\upgauss{\lg \sigma}$~bits to the right might erase a \bsq{1} bit witnessing the rightmost match.
    In the description column, $X$ is treated as an array of integers with bit width~$\upgauss{\lg \sigma}$.
    In this example, $S = \texttt{101010000}$,
  $c$ has the bit representation~\texttt{010} with $\lg \sigma = 3$, and $q = 3$.
  }
  \label{figBroadwordMatching}
\end{figure}

To obtain the number of occurrences of~$c$ in~$S$, we use the $\fnPopcount$ operation returning the number of zero bits in~$X$,
and divide the result by~$\upgauss{\lg \sigma}$.
The $\fnPopcount$ instruction takes \Oh{\lg \lg \lg n} time~\cite[Algo.~1]{vigna08broadword}.
}%
Having \cref{lemBroadwordSearch}, we show that we can compute the frequency of a bigram in~$T$ in \Oh{n \lg \lg \lg n / \log_\sigma n} time.
For that, we partition~$T$ into strings of length~$\gauss{\log_\sigma n}$  fitting into a computer word, and call each string of this partition a \teigi{chunk}.
For each chunk~$S$, we call $\fnFind(c, S)$ to compute the bit vector $X$ storing the occurrences of $c$ in $S$.
In case that we want to use \cref{lemBroadwordSearch} when $c \in \Sigma^2$ is a bigram, we interpret $T \in \Sigma^n$ of length~$n$ as a text $T \in (\Sigma^2)^{\upgauss{n/2}}$ of length~$\upgauss{n/2}$.
The result is, however, not the frequency of the bigram~$c$ in general.
For computing the frequency a bigram~$\texttt{bc} \in \Sigma^2$, we distinguish the cases $\texttt{b} \not= \texttt{c}$ and $\texttt{b} = \texttt{c}$.

\block{Case $\texttt{b} \not= \texttt{c}$}
By applying \cref{lemBroadwordSearch} to find the character~$\texttt{bc}\in\Sigma^2$ in a chunk~$S$ (interpreted as a string of length $\gauss{q/2}$ on the alphabet $\Sigma^2$), we obtain the number of occurrences of~\texttt{bc} starting at odd positions in~$S$.
To obtain this number for all even positions, we apply the procedure to $\texttt{d}S$ with $\texttt{d} \in \Sigma \setminus \{\texttt{b},\texttt{c}\}$.
Additional care has to be taken at the borders of each chunk matching the last character of the current chunk and the first character of the subsequent chunk with \texttt{b} and \texttt{c}, respectively.

\block{Case $\texttt{b} = \texttt{c}$}
This case is more involving as overlapping occurrences of~\texttt{bb} can occur in~$S$, which we must not count.
To this end, we watch out for \teigi{runs} of~\texttt{b}'s, i.e., substrings of maximal lengths consisting of the character~\texttt{b} 
(here, we consider also maximal substrings of~\texttt{b} with length~$1$ as a run).
We separate these runs into runs ending either at even or at odd positions.
We do this because the frequency of~\texttt{bb} in a run of \texttt{b}'s ending at an even (resp.\ odd) position is the number of occurrences of~\texttt{bb} within this run ending at an even (resp.\ odd) position.
We can compute these positions similarly to the approach for $\texttt{b} \not= \texttt{c}$
by first (a) hiding runs ending at even (resp.\ odd) positions, and then (b) counting all bigrams ending at even (resp.\ odd) positions. 
Runs of~\texttt{b} that are a prefix or a suffix of $S$ are handled individually 
if $S$ is neither the first nor the last chunk of~$T$, respectively.
That is because a run passing a chunk border starts and ends in different chunks.
To take care of those runs, we remember the number of~\texttt{b}'s of the longest suffix of every chunk, and 
accumulate this number until we find the end of this run, which is a prefix of a subsequent chunk.
The procedure for counting the frequency of~\texttt{bb} inside~$S$ is explained with an example in \cref{figBroadwordBigram}\JO[~in the appendix]{}.
With the aforementioned analysis of the runs crossing chunk borders, we can extend this procedure to count the frequency of~\texttt{bb} in~$T$.
We conclude:

\begin{lemma}\label{lemBroadwordBigram}
  We can compute the frequency of a bigram in a string~$T$ of length~$n$ whose characters are drawn from an alphabet of size~$\sigma$ 
  in \Oh{n \lg \lg \lg n / \log_\sigma n} time.
\end{lemma}

\begin{figure}[t]
    \newcommand*{\ParboxO}[1]{}
    \newcommand*{\ParboxC}[1]{\parbox[t]{10em}{#1}}
      \centering{%
        \begin{tabular}{lp{3cm}@{\hskip 1em}rl}
      \toprule
      Operation & Description & & Example 
      \\\midrule
      input~$S$ &             & & $\texttt{bbdbbdcbbbdbb} = S$ 
      \\
      $X \gets \fnFind(\texttt{b}, S)$    &
      search \texttt{b} in $S$ & 
                               &
      $\texttt{1101100111011} \rightarrow X$ 
      \\
      $X \gets \fnDeletePrefixRun(X)$ &
      erase prefix of~\texttt{b}'s &
                               &
$\texttt{0001100111011} \rightarrow X$
\\
      $M \gets \fnDeleteSuffixRun(X)$ &
      erase suffix of~\texttt{b}'s &
                               &
$\texttt{0001100111000} \rightarrow M$ 
\\
      $B \gets \fnFindBigram(\texttt{01},M) \fnAND M$ &
      starting of each \texttt{b} run &
                               &
      $\texttt{0001000100000} \rightarrow B$
\\
      $E \gets \fnFindBigram(\texttt{10},M) \fnAND M$ &
      end of each \texttt{b} run & 
                               & 
$\texttt{0000100001000} \rightarrow E$ 
\\
      $M \gets M \fnAND \neg B$ &
      trim head of runs & 
                               &
$\texttt{0000100011000} \rightarrow M$ 
\\ \cdashlinelr{1-4}
      $X \gets B - (E \fnAND (\texttt{01})^{q/2} )$ &
      bit mask for all runs ending at even positions &
      \ParboxO{%
        \phantom{.} \\
        $-$ \\
        \phantom{.} \\
    $=$} &
      \ParboxC{%
        $\texttt{0001000100000}$ = B \\
        $(\texttt{0000100001000}\fnAND$ \\
        \phantom{.}$\texttt{0101010101010})$ \\
    $\texttt{0001000011000} \rightarrow X$}
\\ \cdashlinelr{1-4}
      $X \gets M \fnAND X$ &
      occurrences of all \texttt{b}s belonging to runs ending at even positions &
      \ParboxO{%
        \phantom{.} \\
        $\fnAND$ \\
        $=$} &
      \ParboxC{%
      $\texttt{0001000011000} = X$ \\
      $\texttt{0000100011000} = M$ \\
    $\texttt{0000000011000} \rightarrow X$}
\\ \cdashlinelr{1-4}
      $\fnPopcount(X \fnAND (\texttt{01})^{q/2})$ &
      frequency of all \texttt{bb}s belonging to runs ending at even positions &
      \ParboxO{%
        \phantom{.} \\
        $\fnAND$ \\
        $=$
    } &
      \ParboxC{%
        $\texttt{0000000011000} = X$ \\
        $\texttt{0101010101010}$ \\
      $\texttt{0000000001000}$ }
\\ \cdashlinelr{1-4}
      $X \gets B - (E \fnAND (\texttt{10})^{q/2})$ &
      bit mask for all runs ending at odd positions &
      \ParboxO{%
        \phantom{.} \\
        $-$ \\
        \phantom{.} \\
    $=$} &
      \ParboxC{%
        $\texttt{0001000100000}$ = B \\
        $(\texttt{0000100001000} \fnAND$ \\
        \phantom{.}$\texttt{1010101010101})$ \\
      $\texttt{0000100100000} \rightarrow X$}
\\ \cdashlinelr{1-4}
      $X \gets M \fnAND X$ &
      occurrences of all \texttt{b}s belonging to runs ending at odd positions &
      \ParboxO{%
        \phantom{.} \\
        $\fnAND$ \\
        = 
    } &
      \ParboxC{%
        $\texttt{0000100100000} = X$ \\
        $\texttt{0000100011000} = M$ \\
        $\texttt{0000100000000} \rightarrow X$
    }
\\ \cdashlinelr{1-4}
      $\fnPopcount(X \fnAND (10)^{q/2})$ &
      frequency of all \texttt{bb}s belonging to runs ending at odd positions &
      \ParboxO{%
        \phantom{.} \\
        $\fnAND$ \\
    $=$} &
      \ParboxC{%
      $\texttt{0000100000000} = X$ \\
      $\texttt{1010101010101}$ \\
      $\texttt{0000100000000}$
      }
      \\\bottomrule
    \end{tabular}
      }%
      \caption{
        Finding a bigram \texttt{bb} in a string~$S$ of bit length~$q$, where $q$ is the largest multiple of $2\upgauss{\lg \sigma}$ fitting into a computer word, divided by $\upgauss{\lg \sigma}$.
        In the example, we represent the strings~$M$, $B$, $E$, and $X$ as arrays of integers with bit width~$x := \upgauss{\lg \sigma}$
and write $\texttt{1}$ and $\texttt{0}$ for $1^x$ and $0^x$, respectively.
Let $\fnFindBigram(\texttt{bc}, X) := \fnFind(\texttt{bc}, X) \fnOR \fnFind(\texttt{bc}, \texttt{d}X)$ for $\texttt{d} \neq \texttt{b}$ be the frequency of a bigram~\texttt{bc} with $\texttt{b} \not= \texttt{c}$ as described in \cref{secBroadwordSearch}.
Each of the $\fnPopcount$ queries gives us one occurrence as a result (after dividing the returned number by $\upgauss{\lg \sigma}$), thus the frequency of~\texttt{bb} in~$S$, without looking at the borders of~$S$, is two.
As a side note, modern computer architectures allow us to shrink the $0^x$ or $1^x$ blocks to single bits by instructions like \texttt{\_pext\_u64} taking a single CPU cycle.
      }
      \label{figBroadwordBigram}
    \end{figure}

\subsection{Bit-Parallel Adaption}
Similarly to \cref{lemBatchedCount}, we present an algorithm computing the $d$ most frequent bigrams, but now with the word-packed search of \cref{lemBroadwordBigram}.

\begin{lemma}\label{lemWordPackedCount}
   Given an integer~$d$ with $d \ge 1$,
   we can compute the frequencies of the $d$ most frequent bigrams in a text of length $n$ whose characters are drawn from an alphabet of size~$\sigma$ 
   in \Oh{n^2 \lg \lg \lg n / \log_\sigma n} time using $d \upgauss{\lg(\sigma^2 n/2)} + \Oh{\lg n}$ bits.
\end{lemma}
\begin{proof}
  We allocate a frequency table~$F$ of length~$d$.
  For each text position~$i$ with $1 \le i \le n-1$, we compute the frequency of $T[i]T[i+1]$ in $\Oh{n \lg \lg \lg n / \log_\sigma n}$ time with \cref{lemBroadwordBigram}.
  After computing a frequency, we insert it into~$F$ if it is one of the $d$ most frequent bigrams among the bigrams we have already computed.
  We can perform the insertion in \Oh{\lg d} time if we sort the entries of~$F$ by their frequencies.
\end{proof}

Studying the final time bounds of \cref{eqTotalTime} for the sequential algorithm of \cref{secSequential},
we see that we spend \Oh{n^2} time in the first turn, but spend less time in later turns.
Hence, we want to run the bit-parallel algorithm only in the first few turns until $f_k$ becomes so
large that the benefits of running \cref{lemBatchedCount} outweigh the benefits of the bit-parallel approach of \cref{lemWordPackedCount}.
In detail, for the $k$-th round, we set $d := f_k$ and
run the algorithm of \cref{lemWordPackedCount} on the current text if $d$ is sufficiently small, or otherwise the algorithm of \cref{lemBatchedCount}.
In total, we yield
\begin{align}\label{eqTotalTimeWordPacked}
  \begin{split}
    & {\OhS{\sum_{k=0}^{\Oh{\lg n}} \min\left( \frac{n-f_k}{f_k} n \lg f_k, \frac{(n-f_k)^2 \lg \lg \lg n}{\log_\tau n}}\right) }  \\
    &= {\OhS{n^2 \sum_{k=0}^{\lg n} \min\left(\frac{k}{\gamma^k}, \frac{\lg \lg \lg n}{\log_\tau n}}\right) } \\
    &= {\OhS{\frac{n^2 \lg \log_\tau n \lg \lg \lg n}{\log_\tau n}}} \textup{~time in total,}
  \end{split}
\end{align}
where $\tau = \sigma_m$ is the number of terminals and non-terminals, and
$k / \gamma^k > \lg \lg \lg n/\log_\tau n \Leftrightarrow k = \Oh{\lg (\lg n / (\lg \tau \lg \lg \lg n))}$.

To obtain the claim of \cref{thmWordPack},
it is left to show that the $k$-th round with the bit-parallel approach uses \Oh{n^2 \lg \lg \lg n / \log_\tau n} time, 
as we now want to charge each text position with $\Oh{n/\log_\tau n}$ time with the same amortized analysis as after \cref{eqTotalTime}.
We target \Oh{n / \log_\tau n} time for
\begin{enumerate}[label=(\arabic*)]
  \item replacing all occurrences of a bigram, \label{itBroadWordExchange}
  \item shifting freed up text space to the right,\label{itBroadWordShifting}
  \item finding the bigram with the highest or lowest frequency in~$F$,  \label{itMinMaxFreq}
  \item updating or exchanging an entry in~$F$, and \label{itDynamic}
  \item looking up the frequency of a bigram in~$F$.\label{itFindFreq}
\end{enumerate}
\JO[\Cref{itBroadWordExchange,itBroadWordShifting} can be solved by applying elementary bit-parallel techniques.\footnote{See the full paper at \url{https://arxiv.org/abs/1908.04933} for details.}]{%
Let $x := \upgauss{\lg \sigma_{i+1}}$ and $q$ be the largest multiple of~$x$ fitting into a computer word, divided by~$x$.
For~\cref{itBroadWordExchange}, we partition~$T$ into substrings of length~$q$, and apply \cref{itBroadWordExchange} to each such substring~$S$.
Here, we combine the two bit vectors of \cref{figBroadwordBigram} used for the two 
$\fnPopcount$ calls by a bitwise OR, and call the resulting bit vector~$Y$.
Interpreting~$Y$ as an array of integers of bit width~$x$, $Y$ has $q$ entries, and it holds that
$Y[i] = 2^x-1$ if and only if $S[i]$ is the second character of an occurrence of the bigram we want to replace%
\footnote{like in \cref{itBroadWordExchange}, the case that the bigram crosses a boundary of the partition of~$T$ is handled individually}.
We can replace this character in all marked positions in~$S$ by a non-terminal~$X_{i+1}$ using $x$~bits with the instruction
$(S \fnAND \neg Y) \fnOR ((Y \fnAND L^q) \cdot X_{i+1})$, where $L$ with $|L| = x$ is the bit vector having marked only the least significant bit.
Subsequently, for~\cref{itBroadWordShifting}, we erase all characters $S[i]$ with $Y[i+1] = (Y \ll x)[i] = 2^x-1$ and move them to the right of the bit chunk~$S$ sequentially.
In the subsequent bit chunks, we can use word-packed shifting. 
The sequential bit shift costs $\Oh{|S|} = \Oh{\log_{\sigma_{i+1}} n}$ time, but on an amortized view, a deletion of a character is done at most once per original text position.
}%

For the remaining points,
our trick is to represent $F$ by a minimum and a maximum heap, both realized as array heaps.
For the space increase, we have to lower $\gamma$ adequately.
Each element of an array heap stores a frequency and a pointer to a bigram stored in a separate array~$B$ storing all bigrams consecutively. 
A pointer array~$P$ stores pointers to the respective frequencies in both heaps for each bigram of~$B$.
The total data structure can be constructed at the beginning of the $k$-th round in \Oh{f_k} time, and hence does not worsen the time bounds.
While $B$ solves~\cref{itFindFreq},
the two heaps with~$P$ solve~\cref{itMinMaxFreq,itDynamic} even in \Oh{\lg f_k} time.

In case that we want to store the output in working space, we follow the description in the paragraph after \cref{thmGoal}, 
where we now use word-packing to find the second occurrence of a bigram in $T_i$ in \Oh{n / \log_{\sigma_i} n} time.

\J{%
  \JO[\clearpage{}]{}
\section{Computing MR-Re-Pair in Small Space}\label{secMaxRepeat}
We can adapt our algorithm to compute the MR-Re-Pair grammar scheme proposed by \citet{furuya19repair}.
The difference to Re-Pair is that MR-Re-Pair replaces the most frequent maximal repeat instead of the most frequent bigram,
where a maximal repeat is a reoccurring substring of the text whose frequency\footnote{We naturally extend the definition of \emph{frequency} from bigrams to substrings meaning the number of non-overlapping occurrences.} decreases when extending it to the left or to the right.
Our idea is to exploit the fact that a most frequent bigram corresponds to a most frequent maximal repeat~\cite[Lemma 2]{furuya19repair}.
This means that we can find a most frequent maximal repeat by extending all occurrences of a most frequent bigram to their left and to their right until all are no longer equal substrings.
Although such an extension can be time consuming, this time is amortized by the number of characters that are replaced on creating an MR-Re-Pair rule.
Hence, we conclude that we can compute MR-Re-Pair in the same space and time bounds as our algorithm computing the Re-Pair grammar.

\section{Parallel Algorithm}\label{secParallel}
Suppose that we have $p$ processors on a CRCW machine, 
supporting in particular parallel insertions of elements and frequency updates in a frequency table.
In the parallel setting, we allow us to spend $\Oh{p \lg n}$ bits of additional working space such that each processor
has a extra budget of \Oh{\lg n} bits.
In our computational model, we assume that the text is stored in $p$ parts of equal lengths\footnote{We pad up the last part with dummy characters to match $n/p$ characters.} such that 
we can enlarge a text using $n \lg \sigma$ to $n (\lg \sigma + 1)$ bits in $\max(1, n/p)$ time without extra memory.
For our parallel variant computing Re-Pair,
our working horse is a parallel sorting algorithm:
\begin{lemma}[\cite{batcher68bitonic}]\label{lemSortPar}
  We can sort an array of length~$n$ in $\Oh{\max(n/p,1) \lg^2 n}$ parallel time with \Oh{p \lg n} bits of working space.
  The work is \Oh{n \lg^2 n}.
\end{lemma}
The parallel sorting allows us to state \cref{lemBatchedCount} in the following way:
\begin{lemma}\label{lemBatchedCountPar}
   Given an integer~$d$ with $d \ge 1$,
   we can compute the frequencies of the $d$ most frequent bigrams in a text of length $n$ whose characters are drawn from an alphabet of size~$\sigma$ 
   in \Oh{\max(n,d) \max(n/p,1)  \lg^2 d / d} time using $2d$ $\upgauss{\lg(\sigma^2 n/2)} + \Oh{p \lg n}$ bits.
   The work is \Oh{\max(n,d) n \lg^2 d / d}.
\end{lemma}
\begin{proof}
  We follow the computational steps of \cref{lemBatchedCount}, 
  but (a) divide a scan into $p$ parts, 
  (b) conduct a scan in parallel but a binary search sequentially,
  and (c) use \cref{lemSortPar} for the sorting.
  This gives us the following time bounds for each operation:

  {\renewcommand{\arraystretch}{1.5}
    \setlength{\tabcolsep}{0.5em}
  \begin{tabular}{lll} %
  \toprule
  Operation & \cref{lemBatchedCount} & Parallel \\
  \midrule
  fill $F'$ with bigrams &  ${\Oh{d}}$ & ${\Oh{\max(1,d/p)}}$ \\
  sort $F'$ lexicographically & ${\Oh{d \lg d}}$ & ${\Oh{\max(d/p,1) \lg^2 n}}$ \\
  compute frequencies of $F'$ & ${\Oh{n \lg d}}$ & ${\Oh{n/p \lg d}}$ \\
  merge $F'$ with $F$ & ${\Oh{d \lg d}}$ & ${\Oh{\max(d/p,1) \lg^2 n}}$ \\
  \bottomrule
      \end{tabular}
    }

    The $\Oh{n/d}$ merge steps are conducted in the same way, yielding the bounds of this lemma.
\end{proof}
In our sequential model, we produce $T_{i+1}$ by performing a left shift after replacing all occurrences of a most frequent bigram with a new non-terminal~$X_{i+1}$ such that we gain free space at the end of the text.
As described in our computational model, our text is stored as a partition of $p$ substrings, each assigned to one processor. 
Instead of gathering the entire free space at $T$'s end, we gather free space at the end of each of these substrings.
We bookkeep the size and location of each such free space (there are at most $p$ many) such that we can work on the remaining text~$T_{i+1}$ like it would be a single continuous array (and not fragmented into $p$ substrings).
This shape allows us to perform the left shift in \Oh{n/p} time, while spending \Oh{p \lg n} bits of space for the locations of the free space fragments.

For $p \le n$, exchanging \cref{lemBatchedCount} with \cref{lemBatchedCountPar} in \cref{eqTotalTime} yields 
\[
  \OhS{\sum_{k=0}^{\Oh{\lg n}} \frac{n-f_k}{f_k} \frac{n}{p} \lg^2 f_k} = \OhS{\frac{n^2}{p} \sum_k^{\lg n} \frac{k^2}{\gamma^k}} = \OhS{\frac{n^2}{p}} \textup{~time in total.}
\]
It is left to provide an amortized analysis for updating the frequencies in $F$ during the $i$-th turn.
Here, we can charge each text position with \Oh{n/p} time, as we have the following time bounds for each operation:

  {\renewcommand{\arraystretch}{1.5}
    \setlength{\tabcolsep}{0.5em}
  \begin{tabular}{lll} %
  \toprule
  Operation & Sequential & Parallel \\
  \midrule
  linearly scan $F$ &  ${\Oh{f_k}}$ & ${\Oh{f_k/p}}$ \\
  linearly scan $T_{i}$ & ${\Oh{n_i}}$ & ${\Oh{n_i/p}}$ \\
  sort $D$ with $h = |D|$ & ${\Oh{h \lg h}}$ & ${\Oh{\max(1,h/p) \lg^2 h}}$  \\
  \bottomrule
      \end{tabular}
    }

The first operation in the above table is used, among others, for finding the bigram with the lowest or highest frequency in $F$.
Computing the lowest or highest frequency in $F$ can be done with a single variable pointing to the currently found entry with the lowest or highest frequency during a parallel scan thanks to the CRCW model.\footnote{In the CREW model, concurrent writes are not possible.
A common strategy lets each processor compute the entry of the lowest or highest frequency within its assigned range in $F$, which is then merged in a tournament tree fashion, causing $\Oh{\lg p}$ additional time.}

\begin{theorem}
  We can compute Re-Pair in \Oh{n^2/p} time with $p \le n$ processors on a CRCW machine 
  with $\max( (n/c) \lg n, n \upgauss{\lg \sigma_m}) + \Oh{p \lg n}$ bits of working space including the text space,
  where $c \ge 1$ is a fixed constant, and $\sigma_m$ is the number of terminal and non-terminal symbols.
  The work is \Oh{n^2}.
\end{theorem}

\section{Computing Re-Pair in External Memory}\label{secEM}
The last part of this article is devoted to the first external memory (EM) algorithm computing Re-Pair,
which is another way to overcome the memory limitation problem.
We start with the definition of the EM model, present an approach using a sophisticated heap data structure, 
and another approach adapting our in-place techniques.

For the following, we use the EM model of \citet{aggarwal88iomodel}.
It features fast internal memory~(IM) holding up to $M$ data words, and slow EM of unbounded size.
The measure of the performance of an algorithm is the number of input and output operations (I/Os) required, 
where each I/O transfers a block of $B$ consecutive words between memory levels.
Reading or writing $n$ contiguous words from or to disk requires $\scan(n) = \Ot{n/B}$~I/Os.
Sorting $n$ contiguous words requires $\sort(n)=\Oh{(n/B) \cdot \log_{M/B}(n/B)}$~I/Os.
For realistic values of $n$, $B$, and $M$, we stipulate that $\scan(n) < \sort(n) \ll n$.

A simple approach is based on an EM heap maintaining the frequencies of all bigrams in the text.
A state-of-the-art heap is due to \citet{jiang19heap} providing insertion, deletion, and the retrieval of the maximum element in \Oh{B^{-1} \log_{M/B} (N/B)} I/Os, where $N$ is the size of the heap.
Since $N \le n$, inserting all bigrams takes at most $\sort(n)$ I/Os.
As there are at most $n$ additional insertions, deletions and maximum element retrievals, this sums to at most $4 \sort(n)$ I/Os.
Finally, we need to scan the text $m$ times to replace the occurrences of the retrieved bigram, triggering $m \sum_{i=1}^m \scan(\abs{T_i}) \le m \scan(n)$ I/Os.
In the following, we show an EM Re-Pair algorithm that evades the use of complicated data structures and prioritizes scans over sorting.

This algorithm is based on our Re-Pair algorithm.
It uses \cref{lemBatchedCount} with $d := \Ot{M}$ such that $F$ and $F'$ can be kept in IM\@.
This allows us to perform all sorting steps and binary searches in IM without additional I/O\@.
We only trigger I/O operations for scanning the text, which is done $\upgauss{n/d}$ times, since we partition $T$ into $d$ substrings.
In total, we spend at most $mn/M$ scans for the algorithm of \cref{lemBatchedCount}.
For the actual algorithm, an update of $F$ is done $m$ times, during which we replace all occurrences of a chosen bigram in the text.
This gives us $m$ scans in total. 
Finally, we need to reason about $D$, which is also created $m$ times.
However, $D$ may be larger than $M$, such that we may need to store it in EM\@.
Given that $D_i$ is $D$ in the $i$-th turn, we sort $D$ in EM, triggering $\sort(D_i)$ I/Os.
With a converse of Jensen's inequality~\cite[Theorem B]{simic09jensen} (set there $f(x) := n \lg n$) we obtain
$\sum_{i=1}^m \sort(|D_i|) \le \sort(n) + \Oh{n \log_{M/B} 2}$  total I/Os for all instances of $D$.
We finally yield:

\begin{theorem}
  We can compute Re-Pair with $\min(4\sort(n), (mn/M) \scan(n) + \sort(n) + \Oh{n \log_{M/B} 2} ) + m \scan(n)$~I/Os in external memory.
\end{theorem}
Our approach can be practically favorable to the heap based approach if $m = \oh{\lg n}$ and $mn/M = \oh{\lg n}$,
or if the EM space is also of major concern.
}%

\section{Heuristics for Practicality}\label{secHeuristics} 
The achieved \Oh{n^2} time bound seems to convey the impression that this work is only of purely theoretic interest.
However, we provide here some heuristics, which can help us to overcome the practical bottleneck at the beginning 
of the execution, where only \Oh{\lg n} of bits of working space are available.
In other words, we want to study several heuristics to circumvent the need to 
call \cref{lemBatchedCount} with a small parameter~$d$, as such a case means a considerable time loss.
Even a single call of \cref{lemBatchedCount} with a small~$d$ prevents the computation of Re-Pair of data sets larger than 1 MiB within a reasonable time frame (cf.~\cref{secImplementation}).
We present three heuristics depending on whether our space budget on top of the text space is within
\begin{enumerate}
  \item $\sigma_i^2 \lg n$ bits, \label{itHeurA}
  \item $n_i \lg(\sigma_{i+1} + n_i)$ bits, or \label{itHeurB}
  \item \Oh{\lg n} bits. \label{itHeurC}
\end{enumerate}

\block{Heuristic~\ref{itHeurA}}
If $\sigma_i$ is small enough such that we can spend $\sigma_i^2 \lg n$ bits, 
then we can compute the frequencies of all bigrams in \Oh{n} time. 
Whenever we reach a $\sigma_j$ that lets $\sigma_j \lg n$ grow outside of our budget, 
we have spent \Oh{n} time in total for reaching $T_j$ from $T_i$ as the costs for replacements can be amortized by twice of the text length.

\block{Heuristic~\ref{itHeurB}}
Suppose that we are allowed to use $(n_i-1) \lg (n_i/2) = (n_i-1) \lg n_i - n_i + \Oh{\lg n_i}$ bits additionally to the $n_i \lg \sigma_i$ bits of the text~$T_i$.
We create an extra array $F$ of length $n_i-1$ with the aim that $F[j]$ stores the frequency of $T[j]T[j+1]$ in $T[1..j]$.
We can fill the array in $\sigma_i$ scans over $T_i$, costing us $\Oh{n_i \sigma_i}$ time.
The largest number stored in $F$ is the most frequent bigram in $T$.

\block{Heuristic~\ref{itHeurC}}
Finally, if the distribution of bigrams is skewed, chances are that one bigram outnumbers all others.
In such a case we can use the following algorithm to find this bigram:

\begin{lemma}
   Given there is a bigram in $T_i$ ($0 \le i \le n$) whose frequency is higher than the sum of frequencies of all other bigrams,
we can compute $T_{i+1}$ in \Oh{n} time using \Oh{\lg n} bits.
\end{lemma}
\begin{proof}
   We use the Boyer-Moore majority vote algorithm~\cite{moore91majority} for finding the most frequent bigram in \Oh{n} time with \Oh{\lg n} bits of working space.
\end{proof}

\JO{\subsection*{Acknowledgments}\Acknowledgments{}}

\bibliographystyle{abbrvnat}
\bibliography{literature,references}

\begin{thebibliography}{34}
\providecommand{\natexlab}[1]{#1}
\providecommand{\url}[1]{\texttt{#1}}
\expandafter\ifx\csname urlstyle\endcsname\relax
  \providecommand{\doi}[1]{doi: #1}\else
  \providecommand{\doi}{doi: \begingroup \urlstyle{rm}\Url}\fi

\bibitem[Aggarwal and Vitter(1988)]{aggarwal88iomodel}
A.~Aggarwal and J.~S. Vitter.
\newblock The input/output complexity of sorting and related problems.
\newblock \emph{Commun. {ACM}}, 31\penalty0 (9):\penalty0 1116--1127, 1988.

\bibitem[{Bannai} et~al.(2019){Bannai}, {Hirayama}, {Hucke}, {Inenaga}, {Jez},
  {Lohrey}, and {Reh}]{2019arXiv190806428B}
H.~{Bannai}, M.~{Hirayama}, D.~{Hucke}, S.~{Inenaga}, A.~{Jez}, M.~{Lohrey},
  and C.~P. {Reh}.
\newblock {The smallest grammar problem revisited}.
\newblock \emph{arXiv 1908.06428}, 2019.

\bibitem[Batcher(1968)]{batcher68bitonic}
K.~E. Batcher.
\newblock Sorting networks and their applications.
\newblock In \emph{Proc.\ AFIPS}, volume~32 of \emph{{AFIPS} Conference
  Proceedings}, pages 307--314, 1968.

\bibitem[Bille et~al.(2017{\natexlab{a}})Bille, G{\o}rtz, and
  Prezza]{2017arXiv170408558B}
P.~Bille, I.~L. G{\o}rtz, and N.~Prezza.
\newblock Practical and effective {R}e-{P}air compression.
\newblock \emph{arXiv 1704.08558}, 2017{\natexlab{a}}.

\bibitem[Bille et~al.(2017{\natexlab{b}})Bille, G{\o}rtz, and
  Prezza]{bille17repair}
P.~Bille, I.~L. G{\o}rtz, and N.~Prezza.
\newblock Space-efficient {Re-Pair} compression.
\newblock In \emph{Proc.\ DCC}, pages 171--180, 2017{\natexlab{b}}.

\bibitem[Boyer and Moore(1991)]{moore91majority}
R.~S. Boyer and J.~S. Moore.
\newblock {MJRTY:} {A} fast majority vote algorithm.
\newblock In \emph{Automated Reasoning: Essays in Honor of Woody Bledsoe},
  Automated Reasoning Series, pages 105--118, 1991.

\bibitem[Chan et~al.(2018)Chan, Munro, and Raman]{chan18restore}
T.~M. Chan, J.~I. Munro, and V.~Raman.
\newblock Selection and sorting in the ``restore'' model.
\newblock \emph{{ACM} Trans. Algorithms}, 14\penalty0 (2):\penalty0
  11:1--11:18, 2018.

\bibitem[Charikar et~al.(2005)Charikar, Lehman, Liu, Panigrahy, Prabhakaran,
  Sahai, and Shelat]{charikar05grammar}
M.~Charikar, E.~Lehman, D.~Liu, R.~Panigrahy, M.~Prabhakaran, A.~Sahai, and
  A.~Shelat.
\newblock The smallest grammar problem.
\newblock \emph{{IEEE} Trans. Information Theory}, 51\penalty0 (7):\penalty0
  2554--2576, 2005.

\bibitem[Crochemore et~al.(2015)Crochemore, Grossi, K{\"{a}}rkk{\"{a}}inen, and
  Landau]{crochemore15bwt}
M.~Crochemore, R.~Grossi, J.~K{\"{a}}rkk{\"{a}}inen, and G.~M. Landau.
\newblock Computing the {Burrows-Wheeler} transform in place and in small
  space.
\newblock \emph{J. Discrete Algorithms}, 32:\penalty0 44--52, 2015.

\bibitem[da~Louza et~al.(2017)da~Louza, Gagie, and Telles]{louza17bwt}
F.~A. da~Louza, T.~Gagie, and G.~P. Telles.
\newblock {Burrows-Wheeler} transform and {LCP} array construction in constant
  space.
\newblock \emph{J. Discrete Algorithms}, 42:\penalty0 14--22, 2017.

\bibitem[{De Luca} et~al.(2019){De Luca}, {Russiello}, {Ciro Sannino}, and
  {Valente}]{2019arXiv190110744D}
P.~{De Luca}, V.~M. {Russiello}, R.~{Ciro Sannino}, and L.~{Valente}.
\newblock A study for image compression using {Re-Pair} algorithm.
\newblock \emph{arXiv e-prints}, 2019.

\bibitem[Fredman and Willard(1993)]{fredman93fusion}
M.~L. Fredman and D.~E. Willard.
\newblock Surpassing the information theoretic bound with fusion trees.
\newblock \emph{J. Comput. Syst. Sci.}, 47\penalty0 (3):\penalty0 424--436,
  1993.

\bibitem[Furuya et~al.(2019)Furuya, Takagi, Nakashima, Inenaga, Bannai, and
  Kida]{furuya19repair}
I.~Furuya, T.~Takagi, Y.~Nakashima, S.~Inenaga, H.~Bannai, and T.~Kida.
\newblock {MR-RePair}: Grammar compression based on maximal repeats.
\newblock In \emph{Proc.\ DCC}, pages 508--517, 2019.

\bibitem[Ganczorz(2019)]{ganczorz19entropy}
M.~Ganczorz.
\newblock Entropy lower bounds for dictionary compression.
\newblock In \emph{Proc.\ CPM}, volume 128 of \emph{LIPIcs}, pages 11:1--11:18,
  2019.

\bibitem[Ganczorz and Jez(2017)]{ganczorz17repair}
M.~Ganczorz and A.~Jez.
\newblock Improvements on {Re-Pair} grammar compressor.
\newblock In \emph{Proc.\ DCC}, pages 181--190, 2017.

\bibitem[{Goto}(2017)]{2017arXiv170301009G}
K.~{Goto}.
\newblock Optimal time and space construction of suffix arrays and {LCP} arrays
  for integer alphabets.
\newblock \emph{ArXiv e-prints}, 2017.

\bibitem[Jiang and Larsen(2019)]{jiang19heap}
S.~Jiang and K.~G. Larsen.
\newblock A faster external memory priority queue with decreasekeys.
\newblock In \emph{Proc.\ SODA}, pages 1331--1343, 2019.

\bibitem[K{\"{a}}rkk{\"{a}}inen et~al.(2013)K{\"{a}}rkk{\"{a}}inen, Kempa, and
  Puglisi]{karkkainen13lz77}
J.~K{\"{a}}rkk{\"{a}}inen, D.~Kempa, and S.~J. Puglisi.
\newblock Lightweight {L}empel-{Z}iv parsing.
\newblock In \emph{Proc.\ SEA}, volume 7933 of \emph{LNCS}, pages 139--150,
  2013.

\bibitem[Kieffer and Yang(2000)]{kieffer00code}
J.~C. Kieffer and E.~Yang.
\newblock Grammar-based codes: {A} new class of universal lossless source
  codes.
\newblock \emph{{IEEE} Trans. Information Theory}, 46\penalty0 (3):\penalty0
  737--754, 2000.

\bibitem[Knuth(2009)]{knuthArt4bitwise}
D.~E. Knuth.
\newblock \emph{The Art of Computer Programming, Volume 4, Fascicle 1: Bitwise
  Tricks \& Techniques; Binary Decision Diagrams}.
\newblock Addison-Wesley, 12th edition, 2009.

\bibitem[Larsson and Moffat(1999)]{larsson99repair}
N.~J. Larsson and A.~Moffat.
\newblock Offline dictionary-based compression.
\newblock In \emph{Proc.\ DCC}, pages 296--305, 1999.

\bibitem[Li et~al.(2018)Li, Li, and Huo]{li18sa}
Z.~Li, J.~Li, and H.~Huo.
\newblock Optimal in-place suffix sorting.
\newblock In \emph{Proc.\ SPIRE}, volume 11147 of \emph{LNCS}, pages 268--284,
  2018.

\bibitem[Lohrey et~al.(2013)Lohrey, Maneth, and Mennicke]{lohrey13repair}
M.~Lohrey, S.~Maneth, and R.~Mennicke.
\newblock {XML} tree structure compression using repair.
\newblock \emph{Inf. Syst.}, 38\penalty0 (8):\penalty0 1150--1167, 2013.

\bibitem[Masaki and Kida(2016)]{masaki16repair}
T.~Masaki and T.~Kida.
\newblock Online grammar transformation based on {Re-Pair} algorithm.
\newblock In \emph{Proc.\ DCC}, pages 349--358, 2016.

\bibitem[Navarro and Russo(2008)]{navarro08repair}
G.~Navarro and L.~M.~S. Russo.
\newblock {Re-Pair} achieves high-order entropy.
\newblock In \emph{Proc.\ DCC}, page 537, 2008.

\bibitem[Ochoa and Navarro(2019)]{ochoa19repair}
C.~Ochoa and G.~Navarro.
\newblock {RePair} and all irreducible grammars are upper bounded by high-order
  empirical entropy.
\newblock \emph{{IEEE} Trans. Information Theory}, 65\penalty0 (5):\penalty0
  3160--3164, 2019.

\bibitem[Sakai et~al.(2019)Sakai, Ohno, Goto, Takabatake, I, and
  Sakamoto]{sakai19repair}
K.~Sakai, T.~Ohno, K.~Goto, Y.~Takabatake, T.~I, and H.~Sakamoto.
\newblock {RePair} in compressed space and time.
\newblock In \emph{Proc.\ DCC}, pages 518--527, 2019.

\bibitem[Sekine et~al.(2014)Sekine, Sasakawa, Yoshida, and
  Kida]{sekine14repair}
K.~Sekine, H.~Sasakawa, S.~Yoshida, and T.~Kida.
\newblock Adaptive dictionary sharing method for {Re-Pair} algorithm.
\newblock In \emph{Proc.\ DCC}, page 425, 2014.

\bibitem[Simic(2009)]{simic09jensen}
S.~Simic.
\newblock Jensen's inequality and new entropy bounds.
\newblock \emph{Appl. Math. Lett.}, 22\penalty0 (8):\penalty0 1262--1265, 2009.

\bibitem[Tabei et~al.(2016)Tabei, Saigo, Yamanishi, and Puglisi]{tabei16repair}
Y.~Tabei, H.~Saigo, Y.~Yamanishi, and S.~J. Puglisi.
\newblock Scalable partial least squares regression on grammar-compressed data
  matrices.
\newblock In \emph{Proc.\ SIGKDD}, pages 1875--1884, 2016.

\bibitem[Vigna(2008)]{vigna08broadword}
S.~Vigna.
\newblock Broadword implementation of rank/select queries.
\newblock In \emph{Proc.\ WEA}, volume 5038 of \emph{LNCS}, pages 154--168,
  2008.

\bibitem[Williams(1964)]{william64heapsort}
J.~W.~J. Williams.
\newblock Algorithm 232 - heapsort.
\newblock \emph{Communications of the ACM}, 7\penalty0 (6):\penalty0 347--348,
  1964.

\bibitem[Yoshida and Kida(2013)]{yoshida13repair}
S.~Yoshida and T.~Kida.
\newblock Effective variable-length-to-fixed-length coding via a {Re-Pair}
  algorithm.
\newblock In \emph{Proc.\ DCC}, page 532, 2013.

\bibitem[Ziv and Lempel(1977)]{ziv77lz}
J.~Ziv and A.~Lempel.
\newblock A universal algorithm for sequential data compression.
\newblock \emph{{IEEE} Trans. Information Theory}, 23\penalty0 (3):\penalty0
  337--343, 1977.

\end{thebibliography}

\clearpage
\Jvar{}

\end{document}